\documentclass[a4paper,11pt]{article}

\usepackage{amsmath}

\usepackage{amssymb}
\usepackage{amsthm}
\usepackage{mathrsfs}
\usepackage{wasysym}
\usepackage{bbm}		
\usepackage{mathtools}
\usepackage[shortlabels]{enumitem}
\usepackage{braket}		
\usepackage{mdwlist}		
\usepackage{xcolor}
\usepackage{tikz}
\usepackage{graphicx}
\usepackage[a4paper, left=2.7cm, right=2.7cm, top=3.5cm, bottom=2cm]{geometry}
\usepackage[pdftex=true,hyperindex=true,pdfborder={0 0 0}]{hyperref}
\usetikzlibrary{patterns,positioning,arrows,decorations.markings,calc,decorations.pathmorphing,decorations.pathreplacing}
\usepackage{cleveref}

\usepackage[titletoc]{appendix}

\theoremstyle{plain}
\newtheorem{theorem}{Theorem}[section]
\newtheorem{lemma}[theorem]{Lemma}
\newtheorem{corollary}[theorem]{Corollary}
\newtheorem{proposition}[theorem]{Proposition}
\newtheorem{definition}[theorem]{Definition}

\theoremstyle{definition}
\newtheorem{remark}[theorem]{Remark}
\newtheorem{example}[theorem]{Example}

\newcommand{\infinito}[1]{{}^\infty #1}

\newcommand{\ZZ}{\mathbb{Z}}			
\newcommand{\NN}{\mathbb{N}}			

\usepackage[all]{xy}
\usepackage{color}
\usepackage{microtype}

\newcommand{\e}{\varepsilon}
\renewcommand{\d}{\delta}

\newcommand{\pacman}[1]{\tikz[baseline=.1em,scale=.4]{
\draw (-.1,0) -- (-.1,.85) -- (.8,.85) -- (.8,0) -- cycle; 
  \draw [fill=#1] (.45,.425) -- (.7,.575) arc (+25:+335:.375) -- cycle;
  \fill (0.55,0.6) circle (.35mm)
}}

\newcommand{\cinco}[1]{\tikz[baseline=.0001em,scale=.28]{
  \draw (-.35,-.2) -- (-.35,1.05) -- (1.05,1.05) -- (1.05,-0.2) -- cycle;
  \draw (-.1,-.2) -- (-.1,1.05);
  \draw (.85,-.2) -- (.85,1.05); 
  \draw [fill=#1] (0.05,0.05) -- (.05,.5) arc (+180:0:.3) -- (.65,0.05) --
  (.55,.2) -- (.45,0.05) -- (.35,.2) -- (.25,0.05) -- (.15,.2) -- cycle;
    \coordinate (eye) at (360*rand:.03);
    \foreach \x in {.17,.43}{
      \fill[white] (\x,.5) circle[radius=.1];
      \fill[black] (\x,.5) ++(eye) circle[radius=.05];
    }
}}

\newcommand{\vacio}{\tikz[baseline=.1em,scale=.4]{ 
\draw (-.1,0) -- (-.1,.85) -- (.8,.85) -- (.8,0) -- cycle; 
}}

\newcommand{\ghost}[1]{\tikz[baseline=.1em,scale=.4]{
  \draw (-.1,0) -- (-.1,.85) -- (.8,.85) -- (.8,0) -- cycle; 
  \draw [fill=#1] (0.05,0.05) -- (.05,.5) arc (+180:0:.3) -- (.65,0.05) --
  (.55,.15) -- (.45,0.05) -- (.35,.15) -- (.25,0.05) -- (.15,.15) -- cycle;
    \coordinate (eye) at (360*rand:.03);
    \foreach \x in {.17,.43}{
      \fill[white] (\x,.5) circle[radius=.1];
      \fill[black] (\x,.5) ++(eye) circle[radius=.05];
    }
   \draw [fill=white] (0.12,0.2) -- (0.12,0.3) -- (0.2,0.35) -- (0.28,0.3) -- (0.36,0.35) -- (0.44,0.3) -- (0.52,0.35) -- (0.6,0.3)-- (0.6,0.2) -- (0.52,0.25) -- (0.44,0.2) -- (0.36,0.25) -- (0.28,0.2) -- (0.2,0.25) -- cycle;
}}

\newcommand{\key}[1]{\tikz[baseline=.1em,scale=.4]{
  \draw (-.1,0) -- (-.1,.85) -- (.8,.85) -- (.8,0) -- cycle; 
  \draw [fill=#1] (0.05,0.05) -- (.05,.5) arc (+180:0:.3) -- (.65,0.05) --
  (.55,.15) -- (.45,0.05) -- (.35,.15) -- (.25,0.05) -- (.15,.15) -- cycle;
    \coordinate (eye) at (360*rand:.03);
    \foreach \x in {.17,.43}{
      \fill[white] (\x,.5) circle[radius=.1];
      \fill[black] (\x,.5) ++(eye) circle[radius=.05];
    }
   \draw [fill=white] (0.12,0.2) -- (0.12,0.3) -- (0.2,0.35) -- (0.28,0.3) -- (0.36,0.35) -- (0.44,0.3) -- (0.52,0.35) -- (0.6,0.3)-- (0.6,0.2) -- (0.52,0.25) -- (0.44,0.2) -- (0.36,0.25) -- (0.28,0.2) -- (0.2,0.25) -- cycle;
   \draw [fill=green] (0.1,0.55) arc (+270:0:0.15)  -- (0.45,0.7) -- (0.45,0.8) -- (0.55,0.8) -- (0.55,0.7) -- (0.65,0.7) -- (0.65,0.8) -- (0.75,0.8) -- (0.75,0.55) -- cycle;
}}

\newcommand{\pared}{\tikz[baseline=.1em,scale=.4]{ 
\draw (-.1,0) -- (-.1,.85) -- (.8,.85) -- (.8,0) -- cycle ; 
\draw (.1,0) -- (.1,.85);
\draw (.6,0) -- (.6,.85);
}}

\newcommand{\uno}{\tikz[baseline=.1em,scale=.4]{ 
		\draw (-.1,0) -- (-.1,.85) -- (.8,.85) -- (.8,0) -- cycle; 
		\draw (.1,0) -- (.1,.85);
		\draw (.6,0) -- (.6,.85);
}}

\newcommand{\cero}{\tikz[baseline=.1em,scale=.4]{ 
		\draw (-.1,0) -- (-.1,.85) -- (.8,.85) -- (.8,0) -- cycle; 
}}

\title{Diameter mean equicontinuity and cellular automata}
\author{Luguis de los Santos Ba\~nos and Felipe Garc\'{\i}a-Ramos}

\date{}	

\begin{document}
\maketitle
\begin{abstract}
 Mean and diam-mean equicontinuity are dynamical properties that have been of use in the study of non-periodic order. We show that the Pacman automaton is not almost diam-mean equicontinuous (it is already known that it is almost mean equicontinuous). 	
\end{abstract}
\section{Introduction}
In this paper we study cellular automata (CA) in the context of \emph{topological dynamical systems (TDS)}, i.e., pairs $(X,T)$, where $X$ is a compact metric space and $T:X\rightarrow X$ a continuous function \cite{kuurka2003topological,ceccherini2010cellular}. In particular, we are interested in understanding what types of non-periodic local order can CA exhibit. 


In the context of cellular automata (and symbolic dynamics in general), periodicity is linked to the concept of equicontinuity. A TDS is \emph{equicontinuous} if the family $\{T^n\}_{n\in\NN}$ is equicontinuous. A CA is equicontinuous if and only if it is eventually periodic \cite{kuurka2003topological}. 

One may also study this notion locally. A point $x\in X$ is an \emph{equicontinuity point} if the orbit of a small ball around $x$ will always stay small, that is if for every $\e>0$ there exists $\d >0$ such that 
$diam(T^{i}B_{\d}(x))<\e$ for every $i\in \NN$. We say a TDS is \emph{almost equicontinuous} if the equicontinuity points are dense. Given a CA, it is not difficult to check that a point is equicontinuous if and only if it is locally eventually periodic; that is, if every column is eventually periodic (Proposition \ref{prop:lep}). For CA, almost equicontinuity is more natural than equicontinuity since it can be used to classify CA using sensitivity to initial conditions \cite{kurka1997languages}.

A weaker notion of equicontinuity, \emph{diam-mean equicontinuity}, requires the diameter of small balls to stay small on average (see Definition \ref{def:diam-mean}). The notion of diam-mean equicontinuity has been used to characterize regularity properties of the maximal equicontinuous factor \cite{garciajagerye}, which are natural in the context of aperiodic order. A weaker property, mean equicontinuity, is connected with the concept of discrete spectrum \cite{weakforms,huang2018bounded,downarowiczglasner,fuhrmann2018structure} and almost periodic functions \cite{garcia2019mean} (for a survey on mean equicontinuity see \cite{lisurveymean}). Nonetheless, the view point of this paper is not quite the same as in the study of quasicrystals and aperiodic order as in \cite{baake2013aperiodic}. Because almost mean equicontinuous systems are only required to be ordered locally almost everywhere (not every point is mean equicontinuous), and they may exhibit chaotic properties like positive entropy. 

In \cite{delossantos}, the authors constructed a CA that is almost mean equicontinuous but not almost equicontinuous. In this paper we will show that this CA is not almost diam-mean equicontinuous. The question of whether there exists an almost diam-mean equicontinuous but not almost equicontinuous CA will be addressed in another paper \cite{delossantos2}. 

Sensitivity to initial conditions (or simply sensitivity) is a notion of chaos that can also be studied in mean and diam-mean forms. A cellular automaton that is neither almost mean equicontinuous nor mean sensitive can be constructed \cite{delossantos} (hence Kůrka's dichotomy is not satisfied on the mean forms).

In summary; almost equicontinuous CA exhibit eventually periodic behaviour. CA are either almost equicontinuous or sensitive. Among sensitive CA, there exists very chaotic CA, like the shift, which satisfy all the sensitivity type properties. Nonetheless, among the sensitive CA, there exist almost diam-mean equicontinuous CA and CA that are neither almost diam-mean equicontinuous nor diam-mean sensitive. Among diam-mean sensitive CA, there exists almost mean equicontinuous and CA that are neither mean sensitive nor almost mean equicontinuous. 

We do not know if there are natural hypothesis that imply that a CA must be either almost mean equicontinuous or mean sensitive. 
\section{Preliminaries}
We say $(X,T)$ is a \textbf{topological dynamical system (TDS)} if $X$ is a compact metric space (with metric $d$) and $T:X\rightarrow X$ is a continuous function. Given a metric space $X$, we set $B_{\d}(x)=\{y\in X:d(x,y)<\d\}$, and the diameter of a subset $A$ with $diam(A)$.
A subset of a topological space is \textbf{residual (or comeagre)} if it is the intersection of a countable number of dense open sets.
\begin{definition}

Let $(X,T)$ be a TDS and $x\in X$.  

\begin{enumerate}
\item The point $x$ is an \textbf{equicontinuity point} if for every $\e>0$ there exists $\d >0$ such that $$diam(T^{i}B_{\d}(x))<\e$$ for every $i\in \NN$.  The set of equicontinuity points of $(X,T)$ is denoted by $EQ$.

\item  $(X,T)$ is \textbf{equicontinuous} if 
$EQ=X$.

\item  $(X,T)$ is \textbf{almost equicontinuous} if $EQ$
is a residual set.
 
\end{enumerate}
\end{definition}

The concept of mean equicontinuity first appeared in the work of Fomin\cite{fomin} and Oxtoby\cite{oxtoby}. We say $x$ is a \textbf{mean equicontinuity point} if for every $\e>0$ there exists $\d >0$ such that for every $y \in B_{\d}(x)$ we have that 
$$\limsup_{n\to \infty}\frac{\sum_{i=1}^{n}d(T^{i}x,T^{i}y)}{n}<\e .$$ 

A related concept, diam-mean equicontinuity, was introduced in \cite{weakforms} and studied in \cite{garciajagerye}. For a survey on recent results see \cite{lisurveymean}.

\begin{definition}
	\label{def:diam-mean}
	Let $(X,T)$ be a topological dynamical system. 
	\begin{itemize}
		
		\item We say $x$ is a \textbf{diam-mean equicontinuity point} if for every $\e>0$ there exists $\d >0$ such that 
		$$\limsup_{n\to \infty}\frac{\sum_{i=1}^{n}diam(T^{i}B_{\d}(x))}{n}<\e .$$ 
		We denote the set of diam-mean equicontinuity points by $EQ^{DM}$.
		
		\item $(X,T)$ is \textbf{diam-mean equicontinuous} if $EQ^{DM}=X$.
		
		\item $(X,T)$ \textbf{almost diam-mean equicontinuous} is $EQ^{DM}$ is residual.
		
	\end{itemize}
\end{definition}

It is trivial to see that every equicontinuity point is a diam-mean equicontinuity point, and that every diam-mean equicontinuity point is a mean equicontinuity point.

Now we will give the basic set up of symbolic dynamics. 


 Given a finite set $A$ (called an alphabet), we define the \textbf{$A$-full shift} as $A^{\ZZ }$. If $X$ is the $A$-full shift for some finite $A$ we say that $X$ is a \textbf{full shift}. 
 Given $x\in A^{\ZZ}$, we represent the $i$-th coordinate of $x$ as $x_{i}$. Also, given $i,j\in \ZZ$ with $i<j$, we define the finite word $x_{[i,j]}=x_{i}\dots x_{j}$. Let $A^{+}$ be the set of all finite words.

We endow any full shift with the metric 
\begin{displaymath}
\begin{array}{rcl}
d(x,y) & = & \left\{ \begin{array}{ccl} 2^{-i} & \text{if} \ x\neq y & \text{where} \ i=\min \{ |j| : x_{j}\neq y_{j}  \} ; \\  
0 & \text{otherwise.} &
\end{array} \right.
\end{array}
\end{displaymath}

This metric generates the same topology as the product topology.
For any full shift $A^{\ZZ}$, we define the \textbf{shift map} $\sigma:A^{\ZZ}\rightarrow A^{\ZZ}$ by $\sigma (x)_{i}=x_{i+1}$. The shift map is continuous.


\begin{definition}\label{defCAShift}
We say that $(X,T)$ is a \textbf{cellular automaton (CA)} if $X$ is a full shift and $T:X\rightarrow X$ is continuous and commutes with $\sigma$, i.e., $\sigma \circ T=T\circ \sigma$. 
\end{definition}

\begin{remark}
	Note that $Tx_i$ represents the $i$th coordinate of the point $Tx$, and $Tx_{[0,n]}$ the word extracted from the $[0,n]$-coordinates of the point $Tx$. 
\end{remark}

The following fact can be extracted from the proof of \cite[Theorem 4]{kurka1997languages}.
\begin{proposition}
	\label{prop:lep}
	Let $(X,T)$ be a CA. If $x\in X$ an equicontinuity point then $x$ is locally eventually periodic, i.e., for every $i\in \ZZ$ we have that $T^nx_i$ is an eventually periodic sequence (of $n$). 
\end{proposition}

\section{Pacman CA}


In \cite{delossantos}, the Pacman CA was defined using radius 1 rules.
Let $A=\{ \cero,\uno,\ghost{blue} ,\key{blue},\pacman{yellow},\cinco{blue} \}$. We will define the CA formally and then use a more heuristic explanation of the CA. We define the function $T:A^{\ZZ}\rightarrow A^{\ZZ}$ locally as follows 
{\small
	\begin{displaymath}
		\begin{array}{ccc}
			Tx_{i}& = & \displaystyle\left\{ \begin{array}{ccl}
				\cero & if & (x_{i-1}\in \{ \cero,\ghost{blue},\key{blue} \} \wedge [(x_{i}\in \{ \cero , \ghost{blue}, \key{blue} \} \wedge x_{i+1}\in \{ \cero,\uno,\pacman{yellow} \} )\\ 
				& & \vee ( x_{i}=\pacman{yellow} \wedge x_{i+1}\notin \{ \uno,\cinco{blue} \} )]) \\
				& & \vee (x_{i-1}\in \{ \uno,\cinco{blue} \} \wedge [(x_{i}\in \{ \cero,\key{blue} \} \wedge x_{i+1}\in \{ \cero,\uno,\pacman{yellow} \} ) \\ 
				& & \vee (x_{i}=\pacman{yellow} \wedge x_{i+1}\notin \{ \uno,\cinco{blue} \} )]), \\
				\uno & if & x_{i}\in \{ \uno, \cinco{blue} \} \wedge x_{i+1}\notin \{ \key{blue},\cinco{blue}\} , \\
				
				\ghost{blue} & if & (x_{i-1}\in \{ \cero,\ghost{blue},\key{blue}  \} \wedge x_{i}\in \{ \cero,\ghost{blue},\key{blue}\} \wedge x_{i+1}\in \{ \ghost{blue},\cinco{blue} \} ) \\ 
				&  & \vee (x_{i-1}\in \{ \uno,\cinco{blue} \}\wedge  x_{i}\in \{ \cero,\key{blue} \} \wedge x_{i+1}=\cinco{blue} ), \\
				
				\key{blue} & if & (x_{i+1}=\key{blue}\wedge [(x_{i-1}\in \{ \cero,\ghost{blue},\key{blue}\} \wedge x_{i}\in \{ \cero,\ghost{blue},\key{blue} \} )\\ 
				& & \vee (x_{i-1}\in \{ \uno,\cinco{blue}\} \wedge x_{i}\in \{ \cero,\key{blue} \} )]) \\
				& & \vee (x_{i-1}=\pacman{yellow} \wedge x_{i}\notin \{ \uno,\cinco{blue} \} \wedge x_{i+1}\in \{ \uno,\cinco{blue} \} ) \\
				& & \vee (x_i=\pacman{yellow} \wedge x_{i+1}\in \{ \uno,\cinco{blue} \}),\\
				
				\pacman{yellow} & if &  (x_{i-1}=\uno \wedge [(x_{i}\in \{ \cero,\key{blue} \} \wedge x_{i+1}=\ghost{blue}) \vee x_{i}=\ghost{blue}]) \\
				& & \vee (x_{i-1}= \pacman{yellow} \wedge x_{i},x_{i+1}\notin \{ \uno,\cinco{blue} \}) \text{, and} \\ 
				
				\cinco{blue} & if & x_{i}\in \{ \uno , \cinco{blue} \} \wedge x_{i+1}\in \{ \key{blue},\cinco{blue} \} .\\
				
			\end{array}\right.
			
		\end{array}
	\end{displaymath}
}
This CA has memory and anticipation 1.
We will call the members of the alphabet as follows: 
\begin{itemize}
	\item \cero \ empty space,
	\item \uno \ empty door,
	\item \pacman{yellow} \ pacman,
	\item \ghost{blue} \ ghost.
	\item \key{blue} \ keymaster ghost, and
	\item \cinco{blue} \ door with ghost.
	
\end{itemize}
We will now explain the main properties of this map so the reader gets intuition on the dynamics. The reader does not need to know the rules of the game \emph{Pacman}. One only needs to understand that pacmans eat blue ghosts. 
\begin{itemize}
	\item A door always stays fixed in the same place (a ghost might cross it); that is,  $x_i\in \{\uno, \cinco{blue}\}$ if and only if $Tx_i\in \{\uno, \cinco{blue}\}$.
	
	\item Pacmans \pacman{yellow} move to right (one position per unit of time) if there is no door; that is, if $x_i=\pacman{yellow}$ and  $x_{i+1},x_{i+2}\notin \{\uno, \cinco{blue}\}$ then $Tx_{i+1}=\pacman{yellow}$.
	\item If a pacman encounters a door (on the right) it is transformed into keymaster ghost \key{blue}; that is, if $x_i=\pacman{yellow}$ and $x_{j}\in \{\uno, \cinco{blue}\}$ with $j\in \{i+1,i+2\}$ then $Tx_{j-1}=\key{blue}$.
	\item Ghosts (\ghost{blue},\key{blue}) always move to the left (one position per unit of time) if there is no pacman or a door on the left; that is, if $x_{i}=  \ghost{blue} ( x_{i}=\key{blue})$, $x_{i-1}\in \{ \cero , \ghost{blue}, \key{blue} \}$ and $x_{i-2}\in \{ \cero,  \ghost{blue}, \key{blue}\} (x_{i-2}\neq  \pacman{yellow} )  $, then $Tx_{i-1}=\ghost{blue}(Tx_{i-1}=\key{blue})$. 
	\item If a ghost or keymaster ghost encounters a pacman (on the left) it will disappear (get eaten); that is, 
	\begin{itemize}

		\item if $x_{i}\in \{\ghost{blue}, \key{blue}\}$ and $x_{i-1}=\pacman{yellow}$, then $Tx_{i-1}\notin \{\ghost{blue}, \key{blue}\}$; and 
		
		\item if $x_{i}\in \{\ghost{blue}, \key{blue}\}$, $x_{i-2}=\pacman{yellow}$ and $x_{i-1}\notin \{ \uno ,\cinco{blue} \}$, then $Tx_{i-1}=\pacman{yellow}$. 
		
	\end{itemize}
	
	\item If a ghost \ghost{blue} encounters a door it transforms into a pacman; that is, 
	\begin{itemize}
		\item if $x_{i}=\ghost{blue}$ and $x_{i-1}\in \{ \uno , \cinco{blue}\}$, then $Tx_{i}=\pacman{yellow}$; and
		\item if $x_{i}=\ghost{blue}$, $x_{i-2}=\{ \uno , \cinco{blue}\}$ and $x_{i-1}\notin \{ \uno, \cinco{blue},\pacman{yellow}\}$, then $Tx_{i-1}=\pacman{yellow}$.
	\end{itemize}
	
	\item If a keymaster ghost encounters a door he will enter the door, lose its key, and (in the following step) proceed to the left; that is, if $x_{i}=\key{blue}$ and $x_{i-1}\in \{ \uno , \cinco{blue}\}$, then $Tx_{i-1}=\cinco{blue}$ and
	\begin{itemize}
		\item if $x_{i-3}=\cero$, then $T^{2}x_{i-2}=\ghost{blue}$ and
		\item if $x_{i-3}=\pacman{yellow}$, then $T^{2}x_{i-2}=\key{blue}$.
	\end{itemize} 
	
\end{itemize}
When describing a point in $A^\ZZ$ we will use a point (.) to indicate the \emph{zero}th coordinate, for example if $x= \ ^{\infty}\cero.\pacman{yellow} \ \cero^{\infty}$, it means that $x_0=\pacman{yellow}$ and $x_i=\cero$ for every $i\neq 0$.
We will now provide some examples on how the Pacman CA works. Notice that time flows downward on the diagrams. 

\begin{example}\label{example 2}
	Let $m\geq 2$, and $w=\uno \ \cero^{m} \ \uno$. We will show a section of the orbit of $x:=  ^{\infty }\cero .w\key{blue}^{\infty }$. In this example we can observe that the space between two doors is acting like a some sort of ``filter", because many ghosts disappear. 
	
	\begin{displaymath}
		\begin{array}{cccccccccccc}
			\cero &	\uno        & \cero          &\cero          &\cero          &\cero          &\cero        &\uno        &\key{blue}&\key{blue} \\
			\cero &	\uno        & \cero          &\cero          & \cero         &\cero          &\cero        &\cinco{blue}&\key{blue}&\key{blue} \\
			\cero &	\uno        & \cero          &\cero          &\cero          &\cero          &\ghost{blue}&\cinco{blue}&\key{blue}&\key{blue} \\
			\cero &	\uno        & \cero          &\cero          &\cero          & \ghost{blue} &\ghost{blue}&\cinco{blue}&\key{blue}&\key{blue} \\
			\cero &	\uno        & \cero          &\cero          &\ghost{blue}  &\ghost{blue}  &\ghost{blue}&\cinco{blue}&\key{blue}&\key{blue} \\
			\cero &	\uno        & \cero          &\ghost{blue}  &\ghost{blue}  &\ghost{blue}  &\ghost{blue}&\cinco{blue}&\key{blue}&\key{blue} \\
			\cero &	\uno        & \pacman{yellow}&\ghost{blue}  &\ghost{blue}  &\ghost{blue}  &\ghost{blue}&\cinco{blue}&\key{blue}&\key{blue} \\
			\cero &	\uno        & \cero          &\pacman{yellow}&\ghost{blue}  &\ghost{blue}  &\ghost{blue}&\cinco{blue}&\key{blue}&\key{blue} \\
			\cero &	\uno        & \cero          &\cero          &\pacman{yellow}&\ghost{blue}  &\ghost{blue}&\cinco{blue}&\key{blue}&\key{blue} \\
			\cero &	\uno        & \cero          &\cero          &\cero          &\pacman{yellow}&\ghost{blue}&\cinco{blue}&\key{blue}&\key{blue} \\
			\cero &	\uno        & \cero          &\cero          &\cero          &\cero          &\key{blue} &\cinco{blue}&\key{blue}&\key{blue} \\
			\cero &	\uno        & \cero          &\cero          &\cero          &\key{blue}   &\ghost{blue}&\cinco{blue}&\key{blue}&\key{blue}\\
			\cero &	\uno        &\cero          &\cero          &\key{blue}   &\ghost{blue}  &\ghost{blue}&\cinco{blue}&\key{blue}&\key{blue}\\
			\cero &	\uno        &\cero          &\key{blue}   &\ghost{blue}  &\ghost{blue}  &\ghost{blue}&\cinco{blue}&\key{blue} &\key{blue}\\
			\cero &	\uno        &\key{blue}   &\ghost{blue}  &\ghost{blue}  &\ghost{blue}  &\ghost{blue}&\cinco{blue}&\key{blue}&\key{blue}\\
			\cero &	\cinco{blue}&\pacman{yellow}&\ghost{blue}  &\ghost{blue}  &\ghost{blue}  &\ghost{blue}&\cinco{blue}&\key{blue}&\key{blue}  \\
			
		\end{array}
	\end{displaymath}
	
\end{example}

\begin{example}\label{example 1}
	Let $w=\uno \ \cero \ \cero \ \ghost{blue} \ \pacman{yellow} \ \cero \ \uno \ \key{blue}$. We show a section of the orbit of $x= \ ^{\infty}\cero.w\cero^{\infty}$.
	\begin{displaymath}
		\begin{array}{cccccccccccc}
			\cero & \uno        & \cero          &\cero          &\ghost{blue}  &\pacman{yellow}&\cero        &\uno        &\key{blue}&\cero \\
			\cero & \uno        & \cero          &\ghost{blue}  & \cero         &\cero          &\key{blue} &\cinco{blue}&\cero &\cero \\ 
			\cero & \uno        & \pacman{yellow}&\cero          &\cero          &\key{blue}   &\ghost{blue}&\uno      &\cero & \cero\\  
			\cero & \uno        & \cero          &\pacman{yellow}&\key{blue}   & \ghost{blue} &\cero        &\uno        &\cero &\cero\\
			\cero & \uno        & \cero          &\cero          &\pacman{yellow}&\cero          &  \cero      &\uno        &\cero &\cero\\
			\cero & \uno        & \cero          &\cero          &\cero          &\pacman{yellow}&\cero        &\uno        &\cero &\cero \\
			\cero & \uno        & \cero          &\cero          &\cero          &\cero          &\key{blue} &\uno &\cero &\cero\\
			\cero & \uno        & \cero          &\cero          &\cero          &\key{blue}   &\cero        &\uno &\cero &\cero\\
			\cero & \uno        & \cero          &\cero          &\key{blue}   &\cero          &\cero        &\uno &\cero &\cero\\
			\cero & \uno        & \cero          &\key{blue}   &\cero          &\cero          &\cero        &\uno &\cero &\cero\\
			\cero & \uno        & \key{blue}   &\cero          &\cero          &\cero          &\cero        &\uno &\cero &\cero\\
			\cero & \cinco{blue}& \cero          &\cero          &\cero          &\cero          &\cero        &\uno &\cero &\cero
		\end{array}
	\end{displaymath}
\end{example}

\begin{theorem}
	\cite{delossantos}
	The Pacman CA is almost mean equicontinuous (the set of mean equicontinuity points is dense). 
\end{theorem}

\begin{lemma}[Lemma 3.5 \cite{delossantos}]\label{lem:lemma3}
	Let $(X,T)$ be the Pacman CA, $m>0$, $w\in A^{m}$  and $x= \infinito\vacio .w\pared \ \vacio^{\infty}$. There exists $N>0$ such that for all $n\geq N$,
	\begin{center}
		$T^{n}x_{i}\in \{ \vacio,\pared \} \ \forall \  i\geq 0.$
	\end{center}
\end{lemma}

\begin{lemma}\label{en filita}
Let $(X,T)$ be the Pacman CA, $i>0$ and $w\in A^+$ a finite word. We define the points 
$$x^{i}= \infinito\vacio.w\pared \ \vacio^{i}\key{blue} \ \vacio^{\infty}.$$
There exist $N,M\geq 0$ such that for every $j\in \ZZ_{\geq 0}$ we have that $T^{N+j}x^{M+j}_{0}\in \{\ghost{blue} ,\key{blue},\cinco{blue}\}$. 

\end{lemma}

\begin{proof}
First assume that $w\in \{ \vacio,\pared \}^+$. By simple application of the rules one has that there exists $N_{0}\geq 0$ such that $T^{N_{0}}x_{0}^{1}\in  \{ \ghost{blue}, \key{blue}, \cinco{blue} \} $. Furthermore, since $Tx^i=x^{i-1}$, we obtain the result for $M=0$, that is, for every $j\in \ZZ_{\geq 0}$ we have that $T^{N_{0}+j}x^{j}_{0}\in \{\ghost{blue} ,\key{blue},\cinco{blue}\}$.

For the general case, let $w\in A^{m}$ and $x=\infinito\vacio.w\pared \ \vacio^{\infty}$. Lemma \ref{lem:lemma3} implies that there exists $M'>0$ such that 
$$T^{M'}x_{i} \in \{ \vacio,\pared \} \ \forall \  i\geq 0.$$
Since points $y^{i}:= T^{M'}x^{M'+i}$ look exactly like the cases presented on the first part of the proof, we conclude that there exist $N,M\geq 0$ such that for every $j\in \ZZ_{\geq 0}$ we have that $T^{N+j}x^{M+j}_{0}\in \{\ghost{blue} ,\key{blue},\cinco{blue}\}$.



\end{proof}

We are now ready to prove that the Pacman CA is diam-mean sensitive. Actually it even satisfies a stronger property. 

The strongest form of diameter sensitivity is called cofinitely sensitivity \cite{moothathu2007stronger}. 
\begin{definition}
	Let $(X,T)$ be a TDS. For $U\subseteq X$ and $\d>0$, let $$N_{T}(U,\d):=\{ n\in \NN : \  diam(T^{n}U)>\d \} .$$
	We say that $(x,T)$ is \textbf{cofinitely sensitive} if there exists $\d>0$ such that for every nonempty open set $U\subseteq X$, we have that $N_{T}(U,\d)$ is cofinite (complement is finite).	
	
\end{definition}
It is not difficult to see that every cofinitely sensitive TDS is diam-mean sensitive. 

\begin{theorem}\label{prop:pacman}

The Pacman CA is cofinitely sensitive.

\end{theorem}

\begin{proof}
Let $w$ be a finite word and let $x\in [w]_{0}$ such that $x=\infinito\vacio.w\pared \ \vacio^{\infty}$. Now, let $(y^{i})_{i=0}^{\infty}\subset [w]_{0}$ such that
$$y^{i}=\infinito\vacio.w\pared \ \vacio^{i}\key{blue} \ \vacio^{\infty}.$$
By Lemma \ref{lem:lemma3}, there exists $M\geq 0$ such that for all $n\geq M$ we have that  
$$T^{n} \in \{ \vacio,\pared \} \ \forall \  i\geq 0.$$
Hence, by Lemma  \ref{en filita}, there exists $N\geq 0$ such that
$$d(T^{N+j}x,T^{N+j}y^{M+j})=1$$
for all $j\geq 0$. Therefore, the Pacman CA is cofinitely sensitive. 
\end{proof}

\begin{corollary}
There exists an almost mean equicontinuous CA that is not almost diam-mean equicontinuous. 
\end{corollary}

Using a similar strategy it can be shown that the Pacman Level 2 CA from \cite{delossantos} is also cofinitely sensitive.

Another example that lies within the context of this theory was constructed in \cite{torma2015uniquely}. This CA is non-equicontinuous, uniquely ergodic and the unique measure is supported on a fixed point. By \cite[Theorem 1.2]{garcia2017meanliyorke} this CA must be mean equicontinuous. 

The question of whether there exists a CA that is almost diam-mean equicontinuous but not almost equicontinuous is being addressed in another paper.    
\begin{theorem}
	\cite{delossantos2}
	There exists an almost diam-mean equicontinuous CA that is neither almost equicontinuous nor mean equicontinuous. 
\end{theorem}


\bibliographystyle{plain}
\bibliography{ref}

\medskip

\begin{itemize}
	\item \emph{L. de los Santos Ba\~nos, Instituto de F\'isica, Universidad Aut\'onoma de San Luis Potos\'i, Mexico luguis.sb.25@gmail.com}
	\item \emph{F. Garc\'ia-Ramos, CONACyT \& Instituto de F\'isica, Universidad Aut\'onoma de San Luis Potos\'i, Mexico fgramos@conacyt.mx}
\end{itemize}
\end{document}